\documentclass[a4paper,UKenglish,final]{lipics-v2016}
\pdfoutput=1

\usepackage{amssymb,amsmath,amsfonts,latexsym}
\usepackage{quantum}
\usepackage{xargs}                      % Use more than one optional parameter in a new commands
\usepackage{mathtools}
\usepackage{seqsplit}
\usepackage{rotating}
\usepackage[pdftex,dvipsnames]{xcolor}
\usepackage{multirow}
\usepackage[obeyFinal,colorinlistoftodos]{todonotes}

\newcommand{\ppk}[1]{\todo[inline,color=red!20!white]{#1 ---PPK}}

\newcommand{\tgandhi}[1]{\todo[inline,color=green!20!white]{#1 ---TG}}

\newcommand{\ppkold}[1]{}
\newcommand{\rmittalold}[1]{}
\newcommand{\tgandhiold}[1]{}

\newcommand\mcm[1]{\mathcal{#1}}
\newcommand\mc[1]{\ifmmode \mcm{#1} \else $\mcm{#1}$\fi}
\newcommand\mbm[1]{\mathbb{#1}}
\newcommand\mb[1]{\ifmmode \mbm{#1} \else $\mbm{#1}$\fi}

\newcommand\bhm[1]{\mc{B}({#1})}
\newcommand\bh[1]{\ifmmode \bhm{#1} \else $\bhm{#1}$\fi}
\newcommand{\wt}[1]{\ensuremath{\mathrm{wt}\left(#1\right)}}

\newcommand\chmpm[3]{#1:\bh{#2}\rightarrow \bh{#3}}
\newcommand\chmp[3]{\ifmmode \chmpm{#1}{#2}{#3} \else $\chmpm{#1}{#2}{#3}$\fi}

\newcommand\chosrepm[4]{#1(#2) = \sum_{#3} #4 #2 {#4}^{\dagger}}
\newcommand\chosrep[4]{\ifmmode \chosrepm{#1}{#2}{#3}{#4}  \else $\chosrepm{#1}{#2}{#3}{#4} $\fi}

\newcommand\tensor{\otimes}
\newcommand{\transpose}[1]{\ensuremath{#1^{\mathrm{T}}}}
\newcommand \Cm[1]{\overline{#1}}
\newcommand \C[1]{\ifmmode \Cm{#1} \else $\Cm{#1}$\fi}

\newcommand \Zm[1]{\mc{Z}\left({#1}\right)}
\newcommand \Z[1]{\ifmmode \Zm{#1} \else $\Zm{#1}$\fi}

\newcommand \ordm[1]{\mid {#1}\mid}
\newcommand \ord[1]{\ifmmode \ordm{#1} \else $\ordm{#1}$\fi}
\newcommand \F[1][2]{\ensuremath{\mathbb{F}_{#1}}}

\DeclarePairedDelimiterX\set[1]\lbrace\rbrace{\def\given{\;\delimsize\vert\;}#1}

\newcommand{\Va}{\ensuremath{\mathbf{a}}}
\newcommand{\Vb}{\ensuremath{\mathbf{b}}}
\newcommand{\Vc}{\ensuremath{\mathbf{c}}}
\newcommand{\Vd}{\ensuremath{\mathbf{d}}}
\newcommand{\Vx}{\ensuremath{\mathbf{x}}}
\newcommand{\Vy}{\ensuremath{\mathbf{y}}}
\newcommand{\Vu}{\ensuremath{\mathbf{u}}}
\newcommand{\Vv}{\ensuremath{\mathbf{v}}}
\newcommand{\Vg}{\ensuremath{\mathbf{g}}}
\newcommand{\Vh}{\ensuremath{\mathbf{h}}}
\newcommand{\Vf}{\ensuremath{\mathbf{f}}}
\newcommand{\Vue}{\ensuremath{\mathbf{e_{1}}}}
\newcommand{\Vve}{\ensuremath{\mathbf{e_{2}}}}
\newcommand{\Symplectic}[2]{\ensuremath{\left\langle #1, #2 \right\rangle}}
\newcommand{\pSymplectic}[2]{\ensuremath{\left\langle #1, #2 \right\rangle_\sigma}}
\newcommand{\FFpn}{\mb{F}_p^n \times \mb{F}_p^n}
\newcommand{\ketvec}[1]{\ensuremath{\ket{\mathbf{#1}}}}
\newcommand{\SH}{\ensuremath{N}}
\newcommand{\PR}{\ensuremath{\mc{R}}}
\newcommand{\Sm}{\ensuremath{m}}

\theoremstyle{proposition}
\newtheorem{proposition}[theorem]{Proposition}

\begin{document}

%\title{Stabilizer codes from modified symplectic forms}

%\author{Tejas Gandhi \and Piyush Kurur \and Rajat Mittal }

%\institute{Dept. of Computer Science and Engineering, Indian Institute of Technology Kanpur, Kanpur, UP 208016, India}

% Author macros::begin %%%%%%%%%%%%%%%%%%%%%%%%%%%%%%%%%%%%%%%%%%%%%%%%
\title{Stabilizer codes from modified symplectic form}
\titlerunning{Stabilizer codes from modified symplectic form} %optional, in case that the title is too long; the running title should fit into the top page column

%% Please provide for each author the \author and \affil macro, even when authors have the same affiliation, i.e. for each author there needs to be the  \author and \affil macros
\author[1]{Tejas Gandhi}
\author[2]{Piyush Kurur}
\author[3]{Rajat Mittal}
\affil[1]{Dept. of Computer Science and Engineering, Indian Institute of Technology Kanpur, Kanpur, UP 208016, India \\\texttt{tgandhi@cse.iitk.ac.in}}
\affil[2]{Dept. of Computer Science and Engineering, Indian Institute of Technology Kanpur, Kanpur, UP 208016, India \\\texttt{ppk@cse.iitk.ac.in}}
\affil[3]{Dept. of Computer Science and Engineering, Indian Institute of Technology Kanpur, Kanpur, UP 208016, India \\\texttt{rmittal@cse.iitk.ac.in}}
\authorrunning{T. Gandhi \, P. Kurur \, and R. Mittal} %mandatory. First: Use abbreviated first/middle names. Second (only in severe cases): Use first author plus 'et. al.'

\Copyright{Tejas Gandhi, Piyush Kurur, and Rajat Mittal}%mandatory, please use full first names. LIPIcs license is "CC-BY";  http://creativecommons.org/licenses/by/3.0/

%\subjclass{Dummy classification -- please refer to \url{http://www.acm.org/about/class/ccs98-html}}% mandatory: Please choose ACM 1998 classifications from http://www.acm.org/about/class/ccs98-html . E.g., cite as "F.1.1 Models of Computation". 
% mandatory: Please provide 1-5 keywords
% Author macros::end %%%%%%%%%%%%%%%%%%%%%%%%%%%%%%%%%%%%%%%%%%%%%%%%%

\maketitle

\begin{abstract}
  Stabilizer codes form an important class of quantum error correcting
%codes\cite{gottesman1997stabilizer,ashikhmin2001nonbinary,calderbank1998quantum}
codes which have an elegant theory, efficient error detection, and many
known examples.  Constructing stabilizer codes of length $n$ is
equivalent to constructing subspaces of $\FFpn$ which
are \emph{isotropic} under the symplectic bilinear form defined by
$\Symplectic{(\Va,\Vb)}{(\Vc,\Vd)} = \transpose{\Va} \Vd
%- \transpose{\Vb} \Vc$~\cite{gottesman1997stabilizer,ketkar2006nonbinary,AP02}.
- \transpose{\Vb} \Vc$.
As a result, many, but not all, ideas from the theory of classical
error correction can be translated to quantum error
%correction\cite{ketkar2006nonbinary,cleve1997quantum}. One of the main
correctio. One of the main
theoretical contribution of this article is to study stabilizer codes
starting with a different symplectic form.

In this paper, we concentrate on cyclic codes. Modifying the
symplectic form allows us to generalize the previous known
construction for linear cyclic stabilizer
%codes~\cite{DuttaK2011cyclic,dutta2012exploiting}, and in the process,
code, and in the process,
circumvent some of the Galois theoretic no-go results proved
there. More importantly, this tweak in the symplectic form allows us
to make use of well known error correcting algorithms for cyclic codes
to give efficient quantum error correcting algorithms. Cyclicity of
error correcting codes is a \emph{basis dependent} property. Our codes
are no more \emph{cyclic} when they are derived using the standard
symplectic forms.\footnote{If we ignore the error correcting
properties like distance, all such symplectic forms can be converted
to each other via a basis transformation.}. Hence this change of
perspective is crucial from the point of view of designing efficient
decoding algorithm for these family of codes. In this context, recall
that for general codes, efficient decoding algorithms do not exist
if some widely believed complexity theoretic assumptions are true.

\tgandhi{Should we add reference about hardness of decoding problem, like \cite{iyer2015hardness}?}
\ppk{I added this citation.}

\end{abstract}

%%%%%%%%%%%%%%%%%%%%%%%%%%%%%%%%%%%%%%%%%%%%%%%%%%%%%%%%%%%%%%%

\section{Introduction}

Classical error correcting codes have been instrumental in various areas, 
not just in communication and data storage systems but even in complexity and cryptography. 
In the quantum setting, the major technique to construct error correcting codes is through stabilizers on which 
there exists a substantial body of 
research~\cite{steane1996error,calderbank1996good,rains1999nonbinary,ashikhmin2001nonbinary,calderbank1997quantum,calderbank1998quantum,ketkar2006nonbinary}. 

The theory of quantum information is usually formulated using Hilbert
spaces. Nonetheless, a stabilizer code of block length $n$ over the
$p$-ary alphabet (for some prime $p$) can be uniquely identified with a 
linear subspace $C$ of the space $\F[p]^{2n}$ over the finite field
$\F[p]$. This subspace essentially determines all
the important properties of the code like its distance and dimension
(Theorem~\ref{equivalence_stabilizer_isotropic}) and hence stabilizer
codes can be seen as classical additive codes of twice the block
length. However, for quantum stabilizer codes, the associated subspace
$C$ should be \emph{isotropic}: for any two vectors $\Vu=(\Va,\Vb)$
and $\Vv=(\Vc,\Vd)$ of $C$, the \emph{symplectic linear form}
$\Symplectic{\Vu}{\Vv} = \transpose{\Va}\Vd - \transpose{\Vb}\Vc$
should vanish \cite[Section II]{calderbank1998quantum}\cite[Section
  IV]{ketkar2006nonbinary}. Therefore, constructing quantum stabilizer
codes boils down to constructing an isotropic subspace of $\F[p]^{2n}$
(Theorem~\ref{equivalence_stabilizer_isotropic}). This additional
condition of isotropy is what differentiates quantum codes from
classical codes and often turns out to be a hindrance in transferring
results from classical error correction to the quantum world.

Our main theoretical contribution is to rethink the role played by the
form $\Symplectic{\cdot}{\cdot}$, which was determined by the choice of the
Weyl operators as the basis for quantum errors. The symplectic form
$\Symplectic{\cdot}{\cdot}$ captures the commutation relation between
these Weyl operators. By choosing a different set 
of Weyl operators, any form $\Symplectic{\Vu}{\Vv}_A = \transpose{\Vu} A
\Vv$ can be used for the construction of stabilizer codes\footnote{The matrix 
$A$ needs to be full rank and skew-symmetric for odd prime $p$. The $p=2$ case can be separately handled.}. 
The main idea of this paper is to generalize the
study of stabilizer codes by choosing a different symplectic form as
the starting point. With this change in perspective, we obtain the following results:

\begin{enumerate}
  \item We initiate the study of stabilizer codes based on non-standard symplectic forms. 
  While the freedom to choose the symplectic form is indeed liberating,
  for codes thus constructed, the joint Hamming weight no longer measures the distance of the
  code. Motivated by this difficulty, we formulate the right notion of distance in this context. 
  We identify key features of such symplectic forms required to make bounds on the distance possible.   
  
  \item To complement these theoretical results with concerete examples, we generalize the previous known 
    construction of linear cyclic
    codes by Dutta et. al.~\cite{DuttaK2011cyclic} and are able to work around
    certain no-go theorems proved there. For block lengths $n$ that
    divides $p^t +1$ for some \emph{odd} $t$, Dutta
    et.al.~\cite[Corollary IV.5]{DuttaK2011cyclic} (see also the Ph.D
    thesis \cite{dutta2012exploiting}) proved that there can be no
    linear cyclic stabilizer codes. This impossibility arise due to
    the Galois theoretical restrictions imposed on certain ideals due
    to the isotropy condition arising from the symplectic form
    $\Symplectic{\cdot}{\cdot}$. By modifying the underlying form, we
    are able to circumvent this barrier.

  \item Furthermore, we extend the efficient 
    decoding algorithm given in Dutta et.al.~\cite{dutta2012exploiting,DuttaK2011cyclic}, 
    which in turn uses the celebrated 
    Berlekamp-Massey-Welch~\cite{berlekamp1968algebraic,massey1969-shift-register,welch1986error}
    algorithm for classical cyclic codes. We are able to generalize the 
    decoding algorithm inspite of the fact that our codes are no longer cyclic.
    Notice that the cyclicity of a code is a basis dependent property and our codes are cyclic only 
    when viewed under the modified symplectic form. Thus, these error
    correcting algorithms would not have been possible if we were
    stuck with the standard symplectic form. We believe this is
    important as efficient decoding even for general classical codes
    are intractable~\cite{iyer2015hardness}.
\end{enumerate}

\section{Preliminaries}
%%%%%%%%%%%%%%%%%%%%%%%%%%%%%%%%%%%%%%%%%%%%%%%%%%%%%%%%%%%%%%%

\ppkold{I suggest the following notation through out the writing. We
  should use bold face roman letter like $\mathbf{x}$ for vectors over
  $\F^n$ and the $i$-th component of $\mathbf{x}$ with $x_i$. This
  will make it much clearer. I have added the command
  \texttt{\textbackslash{}ketvec\{x\}} should give you $\ketvec{x}$ }

In the quantum setting, a finite dimensional Hilbert space $\mc{H}$
plays the role of the alphabet. A \emph{quantum block code} $\mc{C}$
of \emph{length} $n$ is just a subspace of the tensor product
$\HS^{\tensor n}$. From now on, we assume that the alphabet space
$\mc{H}$ has a prime dimension $p$. When $p$ is $2$, the Hilbert space
$\mc{H}$ is the space of qubits. We fix an orthonormal basis $\left\{
\ketvec{x} \mid \mathbf{x} \in \F[p] \right\}$ for $\mc{H}$. This is
analogous to picking $\F[p]$ as the alphabet set in the classical
case. Having picked such a basis, a natural basis for the space
$\mc{H}^{\tensor n}$ is given by the set $\left\{ \ketvec{x} \mid \Vx
\in \F[p]^n \right\}$ where $\ketvec{x}$ denotes the state
$\ket{x_1}\tensor \cdots \tensor \ket{x_n}$, where $x_i \in \mb{F}_p$ is the
$i$-th component of $\mathbf{x}$.

For any $\Va$ and $\Vb$ in $\mb{F}_{p}^n$, define the unitary
operators

\begin{equation}\label{shift_operator_equation}
U_{\Va} \ketvec{x} = \ketvec{x+a} \textrm{ and } V_{\Vb}\ketvec{y}=
\omega^{\transpose{\Vb}{\Vy}}\ketvec{y},
\end{equation}
\begin{equation}\label{weyl_commutation_equation}
V_{\Vb}U_{\Va}=\omega^{\Vb^T\Va}U_{\Va}V_{\Vb}
\end{equation}
where $\omega$ is some fixed primitive $p$-th root of unity.  These
operators are called the \emph{Weyl operators} and are used to model
errors in the quantum setting: $U_{\Va}$ corresponds to the \emph{bit
  flips} in the classical setting and $V_{\Vb}$ is the \emph{phase
  flip}. The set of all Weyl operators $U_{\Va} V_{\Vb}$ forms the
basis of the operator space $\mc{B}\left( \HS^{\tensor n} \right)$.
It follows from the general theory of quantum mechanics that any
quantum error in transmission can essentially be modelled using the
Weyl operators. In particular, the group generated by these operators
are what we call the error group.

\begin{definition}\label{error_group}
  Let $p$ be an odd prime. The \emph{error group} $\mc{E}$ associated
  with the block length $n$ is the group of all operators of the from
  $\zeta U_\Va V_\Vb$ where $\zeta$ is a $p$-th root of unity and
  $\Va$ and $\Vb$ are elements of $\F[p]^n$.
\end{definition}

When the characteristic $p$ is $2$, the error group is similar, except that
the scalar factor $\zeta$ is allowed to vary over all the $4$-th roots
of unity, $\{ \pm 1, \pm \iota \}$.

\subsection{Stabilizer Code}

%%%%%%%%%%%%%%%%%%%%%%%%%%%%%%%%%%%%%%%%%%%%%%%%%%%%%%%%%%%%%%%

\rmittalold{Gottesman thesis covers only the case of $\mb{F}_2$. We need
  to find a better reference for the theorems below.}

\tgandhiold{Fixed the Fp citation part}

Stabilizer codes are subspaces that are fixed by some subset $\mc{S}$
of the error group $\mc{E}$. More precisely, for any subset $\mc{S}$,
the subspace

\begin{equation*}
\mc{C}_{\mc{S}} = \left\{ \ket{\psi} \in \HS^{\tensor n} \mid ~ \forall~S \in \mc{S}~~ S\ket{\psi}=\ket{\psi} \right\}
\end{equation*}
is called the stabilizer code associated with the subset
$\mc{S}$. First introduced by Gottesman~\cite{gottesman1997stabilizer}
for the binary alphabet and subsequently
generalized~\cite{rains1999nonbinary,AP02,ketkar2006nonbinary}, the
class of stabilizer codes plays a role analogous to the role played by
linear codes in the classical setting. The following theorem specifies
the conditions under which the code $\mc{C}_{\mc{S}}$ is non-trivial,
 i.e, it has non-zero dimension.

\begin{theorem}\label{stabilizer_group_theorem}
  \cite{gottesman1997stabilizer} For a subset $\mc{S}$ of $\mc{E}$,
  the associated stabilizer code $\mc{C}_{\mc{S}}$ is non trivial if
  and only if
  \begin{enumerate}
    \item $\mc{S}$ forms an Abelian subgroup of the error group $\mc{E}$.
    \item The operator $\zeta I$ \emph{does not} belong to $\mc{S}$
      for any nontrivial root of unity $\zeta$.
  \end{enumerate}

\end{theorem}

A subgroup $\mc{S}$ satisfying the above conditions is called a
\emph{stabilizer subgroup} of the error group.

The \emph{centralizer} $\C{\mc{S}}$ is the set of all operators in
$\mc{E}$ that commute with all the operators of $\mc{S}$. It determines 
the error correcting properties of the
code $\mc{C}_{\mc{S}}$: an error in $\mc{S}$ does not affect the code
space whereas an error in $\mc{E} \setminus \C{\mc{S}}$ leaves a
nontrivial phase on every vector in $\mc{C}_{\mc{S}}$ and hence can be
detected. It is precisely the errors in $\C{\mc{S}} \setminus \mc{S}$
that modifies a vector in $\mc{C}_{\mc{S}}$ to another vector in
$\mc{C}_{\mc{S}}$ and hence cannot be
detected~\cite{gottesman1997stabilizer}. Thus the error correcting
parameters, like the \emph{distance} of the code, depend on the
centralizer.

% Thus this is the forms the set of errors which could not be detected by the quantum code.

Finding a stabilizer subgroup can be reduced to a
problem of designing special subspaces of $\F[p]^n \times \F[p]^n$.
  Given two vectors $\Vu = (\Va,\Vb)$ and
$\Vv=(\Vc,\Vd)$ in $\FFpn $, define the \emph{symplectic
  inner product}, $\Symplectic{\Vu}{\Vv}$, as the scalar
$\transpose{\Va} \Vd - \transpose{\Vb} \Vc$. A subspace $S$ of $\FFpn$
is called \emph{isotropic} if and only if for any two vectors $\Vu$
and $\Vv$ in $S$, $\Symplectic{\Vu}{\Vv} = 0$. From the following theorem,
designing stabilizer subgroup is essentially equivalent to
constructing isotropic subspace.

\begin{theorem} \label{equivalence_stabilizer_isotropic}
  \cite{calderbank1998quantum,AP02,ketkar2006nonbinary}
  \begin{enumerate}
   \item Let $\mathcal{S}$ be a stabilizer subgroup of the error
    group then the subset
    \[ S=\set*{ (\Va,\Vb) \given \zeta U_\Va V_\Vb \in \mc{S}}\]
    %let $S$ be the subset of all pairs $(\Va,\Vb)$  such that
    %$\zeta U_{\Va} V_{\Vb}$ is in $\mathcal{S}$ for some root of unity $\zeta$, then the subset
    is isotropic.
  \item Let $S$ be any isotropic subset of $\FFpn$, then the subgroup
    \[
    \mc{S} = \left\{ \rho(\Va,\Vb) U_{\Va}V_{\Vb} \mid (\Va,\Vb) \in S \subseteq \FFpn \right\}
    \]
    forms a stabilizer subgroup. In the above expression
    $\rho(\Va,\Vb)$ is $\omega^{\frac{1}{2}\Va^T \Vb}$ if $p\neq 2$ and $\iota^{\Va^T \Vb}$ if $p=2$.
  \end{enumerate}
\end{theorem}

The above theorem follows from the fact that for two vectors $\Vu =
(\Va,\Vb)$ and $\Vv = (\Vc,\Vd)$, the Weyl operators $W_{\Vu} =
U_{\Va}V_{\Vb}$ and $W_{\Vv} = U_{\Vc}V_{\Vd}$ commute if and only if
the symplectic inner product $\Symplectic{\Vu}{\Vv} = 0$.  In view of
the above theorem, from now on, stabilizer codes will be characterized
by the associated isotropic subspaces $S$. We also define the
centralizer subspace which corresponds to the centralizer subgroup
$\C{\mc{S}}$.

\begin{definition}[Centralizer subspace]
  Let $S$ be any subspace of $\FFpn$. The \emph{centralizer subspace}
  of $S$, denoted by $\C{S}$, is the subspace of all vectors $\Vu$
  such that $\Symplectic{\Vu}{\Vx} = 0$ for all $\Vx$ in $S$.
\end{definition}

The Hamming weight of an error measures the number of bits that the
error corrupts in the classical setting. For quantum errors, the
\emph{joint weight} is the corresponding measure.

\begin{definition}
  Let $\Vu = (\Va,\Vb)$ be any vector in the vector space $\FFpn$.
  The joint weight $\wt{\Vu}$ is defined as the number of indices $1
  \leq i \leq n$ such that the pair $(a_i,b_i)$ is not $(0,0)$. The joint weight of 
a subset $S$, $\wt{S}$, is the minimum of the joint weights of elements in $S\setminus {0}$. 
\end{definition}

We summarize the error correcting properties of the stabilizer code in
the following theorem~\cite{calderbank1998quantum,AP02}.

\begin{theorem} \label{error_correction_stabilizer}
  Let $S$ be an isotropic subspace of $\FFpn$ and let $\mc{C}_S$ be
  the associated stabilizer code. Then

  \begin{enumerate}
  \item The dimension of $S$ as a vector space over $\F[p]$ is $n - k$
    for some non-negative integer $k$. The dimension of $\mc{C}_S$, as
    a Hilbert space is $p^k$.

  \item If every element of $\C{S}\setminus S$ has joint weight at
    least $d$, then the associated code $\mc{C}_S$ can \emph{detect}
    up to $d-1$ errors and correct up to $\left\lfloor \frac{d-1}{2}
    \right\rfloor$ errors.
  \end{enumerate}
\end{theorem}

Often it is easier to lower bound the distance of the code $\mc{C}_S$ by the joint weight $\wt{\C{S}}$. 
This is known as the \emph{pure distance} of the code.

%The intuition behind defining the joint weight is that it captures the notion of a non-trivial operator.
%In (\ref{joint_weight_equation}), $(a_i,b_i) \neq \left( 0,0 \right)$ implies that there is an non-trivial operator acting on $i$-th qudit.

\rmittalold{we haven't defined the centralizer of a subspace}
\tgandhiold{Added the centralizer as a subspace}

%%%%%%%%%%%%%%%%%%%%%%%%%%%%%%%%%%%%%%%%%%%%%%%%%%%%%%%%%%%%%%%%%%%%%%%%%%%%%%%%%%%%%%%%

%%%%%%%%%%%%%%%%%%%%%%%%%%%%%%%%%%%%%%%%%%%%%%%%%%%%%%%%%%%%%%%%%%%%%%%%%%%%%%%%%
\section{Modifying the symplectic form}
%%%%%%%%%%%%%%%%%%%%%%%%%%%%%%%%%%%%%%%%%%%%%%%%%%%%%%%%%%%%%%%%%%%%%%%%%%%%%%%%%

The isotropy condition associated with the symplectic form
$\Symplectic{\cdot}{\cdot}$ is essentially the only challenge that
prevents us from lifting constructions of classical linear codes to quantum
stabilizer codes
(Theorem~\ref{equivalence_stabilizer_isotropic}). This
symplectic form sometimes even leads to certain Galois theoretic no-go
results~\cite{DuttaK2011cyclic,dutta2012exploiting}. This is what motivates us to modify
the symplectic form and circumvent such impossibility theorems.

Let $p$ be an odd prime and let $A$ be any $2n \times 2n$
skew-symmetric matrix of full rank with entries in $\F[p]$. By
appropriate relabelling of the Weyl operators, the theory of
stabilizer codes can be built where the underlying symplectic form is
given by $\Symplectic{\Vu}{\Vv}_A = \transpose{\Vu} A \Vv$. This is
because, there is always a basis transformation $C$ of $\FFpn$ such
that the $\Symplectic{\cdot}{\cdot}_{\transpose{C}AC}$ is the standard
symplectic $\Symplectic{\cdot}{\cdot}$~\cite[Chapter XV, Corollary
  8.2]{lang:algebra}. The theory of codes then needs to be built out
of the Weyl operators $W_{C\Vu}$ instead of the standard Weyl
operators $W_{\Vu}$. A similar change of symplectic form can be done
in the case when $p$ is 2 as well.

The joint Hamming weight of the vector $(\Va,\Vb)$ measures the number
of indices corrupted by the error $U_\Va V_\Vb$. If we restrict our
attention to symplectic forms given by matrices
$\left(\begin{array}{cc}0&\sigma\\-\sigma^T&0\end{array}\right)$, for
  some $n\times n$ \emph{permutation matrix} $\sigma$, as opposed to
  general forms, a variant of joint Hamming weight would serve the
  purpose of measuring errors -- the weight of $(\Va,\Vb)$ in the
  modified setting should be the joint Hamming weight of $(\Va,\sigma
  \Vb)$, i.e. permute the second component before computing
  weight. Furthermore, if the permutation $\sigma$ in the above matrix
  is also an involution, i.e.  $\transpose{\sigma} = \sigma$, the
  associated symplectic form simplifies further: for vectors
  $\Vu=(\Va,\Vb)$ and $\Vv=(\Vc,\Vd)$, we define the
  \emph{$\sigma$-symplectic inner product} as follows:

  \[ \pSymplectic{\Vu}{\Vv} = \transpose{\Va} \sigma\Vd -  \transpose{\Vb} \sigma\Vc . \]

  The notion of isotropy and centralizer can now be formalized in this
  new setting.

\begin{definition}
  A subspace $S$ of $\FFpn$ is called a \emph{$\sigma$-isotropic
    subspace} if for all $\Vu$ and $\Vv$ $\in$ $S$, $\pSymplectic{\Vu}{\Vv}
  = 0$.

  For any subspace $S$ of $\FFpn$ the \emph{$\sigma$-centralizer}
  $\C{S}$ is the subspace of all vectors $\Vx$ in $\FFpn$ such that
  $\pSymplectic{\Vx}{\Vu} = 0$ for all $\Vu$ in $S$.
\end{definition}

We have the following result that connects standard isotropy and
$\sigma$-isotropy.

\begin{lemma}\label{lem-sigma-standard-connection}

  For any subset $S$ of $\FFpn$, let $S^{\sigma}$ denote the set of
  all elements $(\Va,\sigma\Vb)$ such that $(\Va,\Vb)\in S$, then.
  \begin{enumerate}
  \item $S$ is $\sigma$-isotropic if and only if $S^{\sigma}$ is isotropic
  \item $\C{S}$ is a $\sigma$-centralizer of $S$ if and only if
    $\C{S}^\sigma$ is a centralizer of $S^{\sigma}$.
  \end{enumerate}
\end{lemma}
\begin{proof}
  Since $\sigma$ is an involution we have $\transpose{\sigma} =
  \sigma$. The proof then follows from the identity
  \[ \pSymplectic{(\Va,\Vb)}{(\Vc,\Vd)} = \Symplectic{(\Va,\sigma\Vb)}{(\Vc,\sigma{\Vd})}.
  \]
 \end{proof} 

In view of Theorem~\ref{equivalence_stabilizer_isotropic} and the
previous lemma, it follows that constructing stabilizer codes is
equivalent to constructing $\sigma$-isotropic subspaces of $\FFpn$. We
have the following variant of
Theorem~\ref{error_correction_stabilizer} for $\sigma$-isotropic sets.

\begin{theorem} \label{thm:error_correction_sigma_stabilizer}
  Let $S$ be a $\sigma$-isotropic subspace of $\FFpn$ with
  $\sigma$-centralizer $\C{S}$ then

  \begin{enumerate}
  \item The dimension of $S$ as a vector space over $\FFpn$ is at most $n$, say $n - k$.
        Using $S$, we can construct a stabilizer code of dimension $p^k$.

  \item Suppose, every element of $\C{S}\setminus S$ has joint weight at least $d$,
    then the associated stabilizer code has joint weight at least
    $\left\lfloor \frac{d+1}{2} \right\rfloor$ and correct up
    to $\left\lfloor \frac{d-1}{4} \right\rfloor$ errors.
  \end{enumerate}
\end{theorem}
\begin{proof}
  From Lemma~\ref{lem-sigma-standard-connection}, we have $S^\sigma$
  is isotropic and its centralizer is $\C{S}^\sigma$. Notice that the
  dimension of the space $S^\sigma$ and $S$ are equal as the map
  $(\Va, \Vb) \mapsto (\Va,\sigma \Vb)$ is a permutation on the $2n$
  indices. The stabilizer code required in part 1 is just the
  stabilizer code associated with the isotropic set $S^{\sigma}$.

  Consider any element $(\Va,\Vb)$ and let $A$ denote the indices $i$
  such that $a_i\neq 0$. Similarly let $B$ denote the set of indices
  $j$ such that $b_j \neq 0$. Then the joint weight of $(\Va,\Vb)$
  is the cardinality of $A \cup B$. The joint weight of
  $(\Va,\sigma\Vb)$ is at least the maximum of the cardinalities of
  $A$ and $B$ and hence is at least $\left\lfloor
  \frac{d+1}{2}\right\rfloor$. It follows that $\C{S}^\sigma \setminus
  S^\sigma$ has joint weight at least $\left\lfloor
  \frac{d+1}{2}\right\rfloor$. Using
  Theorem~\ref{error_correction_stabilizer} we get the necessary
  result.
 \end{proof}

%Notice that the previous theorem is more or less equivalent to Theorem~\ref{error_correction_stabilizer}. 
The bound on the distance in the previous theorem is conservative. Theoretically this is the
best bound that we can derive. Explicit examples
constructed in Section \ref{section_code_examples} often give
much better distances. It may as well be the case that the actual joint weight of
$\C{S}^\sigma \setminus S^\sigma$ could even be higher than that of $\C{S}
\setminus S$.

\subsection{Cyclic codes}
%%%%%%%%%%%%%%%%%%%%%%%%%%%%%%%%%%%%%%%%%%%%%%%%%%%%%%%%%%%%%%%%%%%%%%%%%%%%%%%%%%%%%%%

We fix a finite field $\F[p]$ as the alphabet set and a block length
$n$ that is co-prime to $p$. Consider the right shift operator
$\SH$ that maps a vector $\Va = (a_0,\ldots, a_{n-1})$ to its right shift
$(a_{n-1},a_0,\ldots,a_{n-2})$. A classical code $C$ is \emph{cyclic}
if for all $\Va$ in $C$ its right shift $\SH \Va$ is also in $C$. It
turns out that the right generalization of this notion is \emph{simultaneous
cyclicity}.

\begin{definition}\label{simulcyc_definition}
A subset $S$ of $\FFpn$ is \emph{simultaneously cyclic} if for all
$(\Va,\Vb)$ in $S$, its simultaneous shift $(\SH \Va,\SH \Vb)$ is also
in $S$.
\end{definition}

A quantum stabilizer code is \emph{cyclic} if the associated isotropic
set $S$ is simultaneously cyclic~\cite[III.2]{DuttaK2011cyclic}. It
turns out that the centralizer $\C{S}$ is also simultaneously cyclic.
As in the case of classical cyclic codes, the associated code can be 
seen as an ideal over an appropriate cyclotomic ring.  In the more general setting of
$\sigma$-isotropic sets, for a simultaneously cyclic subspace $S$, its
centralizer $\C{S}$ \emph{need not} be simultaneously cyclic and hence
the theory of cyclotomic rings  will not be applicable any
more. Though, if we further restrict to involution $\sigma$ to be of the form $i
\mapsto m i$ modulo $n$ for some $m$, we get back all the nice
properties that we are accustomed to in the classical setting. 
We will call such an involution $\sigma_m$, fix it
for the rest of the article. Notice, $\sigma_m$
being an involution means $m$ is a square root of $1 \mod n$.  It is
easy to see that the shift operator $\SH$ and $\sigma_m$ satisfy the
commutation relation:

\begin{equation}\label{eq:shift-sigma-commute}
   N \sigma_\Sm  = \sigma_\Sm \SH ^\Sm
\end{equation}

The following theorem  follows directly.

\begin{theorem} \label{centralizer_sim_cyc_theorem}
  Let $S$ be a $\sigma_\Sm$-isotropic, simultaneously cyclic subspace
  of $\FFpn$. Then its $\sigma_\Sm$-centralizer $\C{S}$ is also
  simultaneously cyclic.
\end{theorem}

%% Define $\sigma_\Sm \in S_n$ to be the permutation, $\sigma_\Sm(i) =\Sm i \mod n$.
%% Since we are only interested in involutions, the only possible values of $\Sm $ are the square-roots of $1$ in $\mb{Z}_n$.
%% Such a $\sigma_\Sm$ satisfy \emph{approximate commutation relation} between the shift operator, $\SH  \in S_n$, and $\sigma_\Sm$.

%% \begin{eqnarray*}
%% \sigma_\Sm (\SH ^\Sm (i) )&=& \sigma_\Sm (\SH ^\Sm (i))\\
%%            &=& \sigma_\Sm (i+\Sm )\\
%%            &=& \Sm i+1\\
%%            &=& (\sigma_\Sm (i)+1)\\
%%            &=& \SH (\sigma_\Sm (i))
%% \end{eqnarray*}
%% As an immediate consequence of this we have the following theorem.

Consider the cyclotomic ring $\PR =\F[p][X]/\left<X^n-1\right>$. As in
the classical case, representing a vector $\Va$ as the polynomial
$\Va(X)=a_0+a_1X+\cdots+a_{n-1}X^{n-1} \in \PR$ provides an elegant
mechanism to work with cyclic codes. For example, the cyclic shift of
a vector in $\F[p]^n$ is equivalent to multiplication by $X$ in the
ring $\PR$. The following theorem expresses the $\sigma_\Sm $-isotropy
condition in terms of these polynomial representations.

\begin{theorem}\label{permuted_isotropic_subspace_theorem}
  Let $S$ be a simultaneously cyclic subspace of $\FFpn$. $S$ is a
  $\sigma_\Sm $-isotropic subspace if and only if for any two elements
  $\Vu=\left( \Va,\Vb \right)$ and $\Vv= \left( \Vc,\Vd \right)$ in
  $S$, the corresponding polynomials satisfy the identity:
  \begin{equation}\label{polynomial_isotropy_condition}
    \Va(X)\Vd(X^{-\Sm}) - \Vb(X)\Vc(X^{-\Sm}) = 0 \mod X^n -1
  \end{equation}
\end{theorem}

\begin{proof}
The constant coefficient of the polynomial on the left hand side of
Equation~\ref{polynomial_isotropy_condition} is equivalent to the
$\sigma_m$-isotropy condition. So the polynomial condition implies
that the subspace $S$ is $\sigma_m$-isotropic.

For the converse, notice that $\Va(X^\Sm) = (\sigma_\Sm
\Va)(X)$. Hence it is sufficient to prove the identity,
\[
\Va(X)(\sigma_\Sm \Vd)(X^{-1}) - \Vb(X)(\sigma_\Sm \Vc)(X^{-1}) = 0 \mod
X^n -1,
\]
for every $(\Va,\Vb), (\Vc,\Vd) \in S$.
The coefficient of $X^i$ on the left is,
\[ \transpose{\Va}(\SH ^i \sigma_\Sm) \Vd - \Vb^T (\SH ^i \sigma_\Sm) \Vc. \]
Using the Equation~\ref{eq:shift-sigma-commute} repeatedly, the coefficient simplifies to,
\[ \transpose{\Va}(\sigma_m \SH^{mi}) \Vd - \Vb^T (\sigma_m \SH^{mi})\Vc, \]
which is equal to $\Symplectic{(\Va,\Vb)}{(N^{mi}\Vc,N^{mi}\Vd)}_{\sigma_m}$.

Since $S$ is $\sigma_m$-isotropic and simultaneously cyclic, this coefficient is $0$.

 \end{proof} 

%% Recall that~\cite{some-std-reference-classical-cyclic} any
%% classical cyclic code is then an ideal of the ring $\PR$. For any such
%% ideal $C$, there is a factor $g(X)$ of $X^n - 1$ that divides all the
%% polynomials in $C$ and this polynomial is called the \emph{generator
%%   polynomial} of $C$.

\subsection{Linear codes}
%%%%%%%%%%%%%%%%%%%%%%%%%%%%%%%%%%%%%%%%%%%%%%%%%%%%%%%%%%%

Let $p$ be a prime. It is well known that the finite field $\F[p]$ has
a \emph{unique} quadratic extension $\F[p^2]$.  Such an extension is
essentially the field $\F[p](\eta)$ consisting of all elements of the
form $a + \eta b$, where $\eta$ is the root of some irreducible
quadratic polynomial $\mu(Y) = Y^2-c_1Y-c_0$.  The encoding $(\Va,\Vb)
\mapsto \Va + \eta \Vb$ gives an encoding of $\FFpn$ to
$\F[p^2]^n$. We fix such an element for the rest of the section.
A $\F[p^2]$-linear subspace of $\F[p^2]^n$ is isotropic iff the
corresponding preimage forms an isotropic subspace of $\FFpn$.
Stabilizer codes associated to isotropic $\F[p^2]$-vector spaces 
are called \emph{linear stabilizer codes}. As isotropic subspace
is an $\F[p^2]$-linear subspace, it is closed under the multiplication
by $\eta$.

For an $\F[p]$ subspace $C$ of $\FFpn$, the necessary and sufficient
condition for it to be an $\F[p^2]$-subspace under the above encoding
is that it should be closed under multiplication by $\eta$. Since
$\eta^2 = c_o + c_1 \eta$, $C$ is $\F[p^2]$-linear iff 
for all pair $(\Va,\Vb)$ in $C$, the pair $(c_o
\Vb, \Va + c_1 \Vb)$ also belongs to $C$.  We have the following
theorem on $\sigma$-centralizers.

\begin{theorem}\label{linearity_theorem}
  Let $\sigma$ be an involution in $S_n$ and $S$ be a
  $\sigma$-isotropic subspace. Let $\C{S}$ be the corresponding
  $\sigma$-centralizer subspace. Then $S$ is $\mb{F}_{p^2}$-linear
  implies $\C{S}$ is $\mb{F}_{p^2}$-linear.
\end{theorem}
\begin{proof}
  Consider any arbitrary element $\Vv=\left( \Vx,\Vy \right)$ in
  $\C{S}$. By the definition of $\sigma$-centralizer, for any
  $\Vu=(\Va,\Vb)$ in $S$, we have $\pSymplectic{\Vu}{\Vv} =0$. In
  addition, we also have
  \begin{equation}\label{psymplectic_eta_u_v_equation}
    \pSymplectic{\eta\Vu}{\Vv} = c_0\Vb^{T}\sigma\Vy - \transpose{\Va}\sigma\Vx - c_1\transpose{\Vb}\sigma\Vx = 0.
  \end{equation}
  due to the linearity of $S$. Now consider $\pSymplectic{\Vu}{\eta
    \Vv}$. We have:
  \begin{align*}
    \pSymplectic{\Vu}{\eta\Vv} &=& \Va^{T}\sigma\left( \Vx+c_1\Vy
    \right) -\Vb^{T}\sigma\left( c_0\Vy \right) && \\ &=& \Va^{T}
    \sigma\Vx +c_1\Va^{T}\sigma\Vy -c_0\Vb^{T}\sigma\Vy && \\ &=& -
    c_1\transpose{\Vb} \sigma \Vx + c_1 \transpose{\Va} \sigma\Vy &&
    \text{ from (\ref{psymplectic_eta_u_v_equation})} \\ &=& c_1
    \pSymplectic{\Vu}{\Vv}\\ &=& 0
  \end{align*}
  This proves that $\eta \Vv$ is in $\C{S}$ and hence $\C{S}$ is a
  $\F[p^2]$-subspace.
 \end{proof}

The above theorem is true for any involution $\sigma$. In addition, if
the involution is $\sigma_\Sm$, using Theorem
~\ref{centralizer_sim_cyc_theorem} we have

\begin{theorem}
  Let $S$ be a simultaneously cyclic, $\sigma_\Sm$-isotropic and
  $\mathbb{F}_{p^2}$-linear subspace. Then its
  $\sigma_\Sm$-centralizer $\C{S}$ is simultaneously cyclic and
  $\mathbb{F}_{p^2}$-linear subspace.
\end{theorem}

We now look at $\F[p^2]$-linear, $\sigma_\Sm$-isotropic, and simultaneously cyclic
 subspaces. Recall that we encoded a vector $\Va$ as the polynomial 
$\sum a_i X^i$ in the ring $\PR$. In this  case, we encode pairs
 of polynomials in $\PR \times \PR$ as elements of 
$\PR(\eta) = \F[p^2][X]/\left< X^n-1\right>$. A $\F[p^2]$-linear
 simultaneously cyclic subspace  has to be an 
ideal of $\PR(\eta)$ and hence should have a generating polynomial.
 A consequence of previous theorem is that its centralizer would also
 be an ideal and have a generating polynomial.

\subsection{Linear stabilizer codes from $\sigma_{\Sm}$-isotropic sets} \label{sec-cyclic-sigma-m1-isotorpic}
%%%%%%%%%%%%%%%%%%%%%%%%%%%%%%%%%%%%%%%%%%%%%%%%%%%%%%%%%%%%%%%%%%

An $\F[p^2]$-linear simultaneously cyclic subspace is equivalent to an ideal of $\PR(\eta)$.
Define a triplet $(n,p,m)$ to be \emph{good} if $n \divides p^{t} + m$ for some $t$ and $m^2 = 1 \mod n$. 
The following theorem characterizes $\sigma_\Sm$-isotropic ideals of $\PR(\eta)$ for good triplets.

%The following theorem generalizes the codes constructed by  Dutta et.al.~\cite{DuttaK2011cyclic,dutta2012exploiting} constructed
%quantum linear cyclic codes for length $n$ and alphabet size $p$ if $n
%\divides p^t+1$ for some even $t$. Furthermore, they prove that if $n$
%divides $p^t + 1$ for some odd $t$ then there are no linear cyclic
%stabilizer codes. We now show how considering $\sigma_\Sm$-isotropic
%sets, we can circumvent the restriction.
%
% Recall that all cyclic linear
%stabilizer codes corresponds to the ideals of $\PR(\eta)$. We have a
%complete characterization of such ideals.

\begin{theorem}\label{general_construction_theorem}

  Let $n \divides p^{t} + m$, where $m$ is a square root of
  $1 \mod n$. An ideal $S$ of $\PR(\eta)$ is $\sigma_m$-isotropic 
  if and only if $S$ is generated by 
  the product of two polynomials
  $g(X)h(X,\eta)$ which satisfy the following conditions.

  \begin{enumerate}
  \item $g(X)$ is a factor of $X^n - 1$ in $\F[p]$ which includes all
    the odd irreducible factors.
  \item $h(X,\eta)$ is such that for any $r(X,\eta)$ which is a factor
    of $(X^n-1)/g(X)$ over $\mathbb{F}_{p}(\eta)$ exactly one of
    $r(X,\eta)$ or its conjugate $r(X,\eta)^p$ divides $h(X,\eta)$.
\end{enumerate}
 Hence the ideal will be non-trivial only if $t$ is even.
\end{theorem}
\begin{proof}
  Refer to the section \ref{proof_construction_theorem} of the appendix.
\end{proof} 

Dutta et. al.~\cite{DuttaK2011cyclic,dutta2012exploiting}
gave a characterization of the $\sigma_\Sm$-isotropic ideals of $\PR(\eta)$ for 
good triplets where $m=1$. Theorem~\ref{general_construction_theorem} generalizes there 
construction for other square roots of $1 \mod n$.  

Dutta et. al.~\cite{DuttaK2011cyclic} also proved that if $n \divides p^t + 1$ for some odd $t$ 
then there are no linear cyclic stabilizer codes. This will be the case when the order of $p$ in $\ZZ_n$ is 
$2t$ (for some odd $t$) and $p^{t} = -1 \mod n$. However, Theorem.~\ref{general_construction_theorem} 
allows us to construct $\sigma_{-1}$-isotropic ideals for such $(n,p)$ pairs.

For $m \neq -1$, the $\sigma_\Sm$-isotropy condition as polynomials is
given by
\begin{equation}\label{eqn:standard-isotropy}
\Va(X)\Vd(X)^{p^t} - \Vb(X)\Vc(X)^{p^t} \mod X^n -1
\end{equation}
The $p^t$ powers that occur in the above condition leads to certain
Galois theoretic situations that makes such ideals trivial when $t$ is
odd. The $\sigma_{-1}$-isotropy condition on the other hand is much
simpler
\[
\Va(X)\Vd(X) - \Vb(X)\Vc(X) \mod X^n -1.
\]
This is the reason why we are able to construct linear cyclic
stabilizer codes when we modify the symplectic condition to the
bilinear form $\Symplectic{\cdot}{\cdot}_{\sigma_{-1}}$.

Our results also gives other variants: If $p$ has order $4t$ and
the quantity $-m = p^{2t}$ is a square root different from $\pm 1$, by
considering $\sigma_{\Sm}$-isotropic sets we get examples of codes that
were not considered above. Such non-trivial square roots exists when
the block length $n$ is composite. With such variants, we may be
able to prove better lower bounds than the one mentioned here (Theorem~\ref{thm:error_correction_sigma_stabilizer}).

When the order of $p$ in $\ZZ_n$ is odd, all the factors of $X^n-1$
have odd degree. Therefore, we do not have any non-trivial $\sigma_{\Sm}$-isotropic
ideals.

To summarize:
\begin{itemize}
\item Prime $p$ has order $4t$ and $p^{2t} = -1$ then
  construct codes based on the work of Dutta et. al. \cite{DuttaK2011cyclic,dutta2012exploiting}.
\item Prime $p$ has even order, construct codes based on any one of
  the non-trivial $\sigma_{-1}$-isotropic ideals using
  Theorem~\ref{general_construction_theorem}.
\item Prime $p$ has odd order in $\ZZ_n$, our strategy fails.
\end{itemize}

\tgandhi{ The codes constructed by Theorem Frobenius construction theorem are subsumed by the construction of Theorem \ref{general_construction_theorem}. With this respect the Theorem Frobenius construction theorem appears of no use. }

\ppk{That is incorrect, $p^t = m$ can be a square root of 1 that is different from $\pm 1$}

\tgandhi{ That true. But still Theorem Frobenius construction theorem 
 does not provide the ability to construct more codes. 
 All the codes created by Theorem Frobenius construction theorem 
can be constructed by Theorem \ref{general_construction_theorem}.
 As both Theorem \ref{general_construction_theorem} and 
 Theorem Frobenius construction theorem impose the same constraint on
 the ideal generating them. Moreover these constraint are
 are a consequence of uniquely cyclic and 
$\F[p^2]$-linearity. The only part where $\sigma_\Sm$-isotropy
 plays a role is in enforcing $t$ to be even. 
 See Section 5.1 for the details }

In the constructions that we have sketched above, we need a
characterization of the $\sigma$-centralizer $\C{S}$ if we need some
handle on the error correction properties. We have the following
proposition (proof in section \ref{proof_centralizer-generator} of the
appendix).

\begin{proposition}\label{prop:centralizer-generator}
  Consider a $\sigma$-isotropic ideal $S$ whose generating polynomial
  is $g(X) h(X,\eta)$ as in
  Theorem~\ref{general_construction_theorem}.  The
  $\sigma$-centralizer $\C{S}$ of $S$ is given by the ideal generated
  by $h(X,\eta)$.
\end{proposition}

\ppk{The initial wording was "$\C{S}$ maps to an ideal" which clearly looked odd.}

\tgandhi{$h(X,\eta)$ does not generate the centralizer $\C{S}$ under the map
 $(\Vu,\Vv) \mapsto \Vu+\eta \Vv$. However, centralizer $\C{S}$ maps to the ideal
$h(X,\eta)$ under the map $(\Vu,\Vv) \mapsto \Vu+\eta c_0^{-1} \Vb$.
That is the reason to add the word "$\C{S}$ maps to an ideal" without specifying to
much details here. }

%\input{distance}
%%%%%%%%%%%%%%%%%%%%%%%%%%%%%%%%%%%%%%%%%%%%%%%%%%%%%%%%%%%%%%
\section{BCH distance and decoding}\label{sec-bch-distance-and-decoding}
%%%%%%%%%%%%%%%%%%%%%%%%%%%%%%%%%%%%%%%%%%%%%%%%%%%%%%%%%%%%%%

For the generators of cyclic codes, we now define the BCH-distance.

\begin{definition}\label{ideal_distance_proposition}
  Let $f(X)$ be any factor of $X^n - 1$ over a field $\F[q]$. The BCH
  distance of $f(X)$ is the maximum $l$ such that there exists a
 sequence of $l-1$ consecutive powers $\beta,\beta^2,\ldots,\beta^{l-1}$,
  all of which are roots of $f(X)$ for some primitive $n$-th root of
  unity $\beta$.
\end{definition}

The BCH distance of $f(X)$ gives a lower bound on the distance of the
associated ideal as a code. Notice that for the codes constructed in
the previous section, the centralizer as an ideal over $\PR[\eta]$ has
generator $h(X,\eta)$ and the Hamming distance of this code is its
joint weight in the quantum setting. Using
Theorem~\ref{thm:error_correction_sigma_stabilizer}, we have the
following proposition.

\begin{proposition}
  Let $S$ be a linear cyclic stabilizer code generated using one of
  Theorems~\ref{general_construction_theorem} and let $g(X)h(X,\eta)$ be
  its generator polynomial. If the polynomial $h(X,\eta)$ has BCH
  distance $d$, then the associated stabilizer code has distance at-least
  $\lfloor\frac{d+1}{2}\rfloor$.
\end{proposition}

For a classical cyclic code, the celebrated
Berlekamp-Massey-Welch~\cite{berlekamp1968algebraic,massey1969-shift-register,welch1986error}
algorithm gives an efficient error correction procedure. We
reformulate this result for our use in the decoding of quantum cyclic codes.

\begin{theorem}[Berlekamp-Massey-Welch~\cite{berlekamp1968algebraic,massey1969-shift-register,welch1986error}]
\label{berlekamp_theorem}
  Let $f(X)$ be any factor of $X^n-1$ over a finite field $\F[q]$ with
  BCH distance at least $2t+1$. Let $e(X)$ be any unknown polynomial
  with Hamming weight at most $t$. There exists an efficient
  algorithm, that takes as input any $r(X)=e(X) \mod f(X)$ and outputs $e(X)$.
\end{theorem}

Let $S$ be any $\sigma_\Sm$-isotropic cyclic code that we constructed
in the previous section. We fix some notation for this section.  Recall,
$S^{\sigma_\Sm}$ is the space of all elements $(\Va,\sigma_\Sm \Vb)$,
where $(\Va,\Vb)$ is in $S$. This set forms an isotropic subset under the
standard symplectic form
(Lemma~\ref{lem-sigma-standard-connection}). The corresponding
stabilizer group $\mc{S}'$ consists of Weyl operators $W_{\Va,\Vb}' =
\rho(\Va,\sigma_m \Vb)U_\Va V_{\sigma_\Sm \Vb}$, $\left(\Va,\Vb
\right)$ is in $S$, and the quantum code $\mc{C}$ is the space of
vectors stabilized by $\mc{S}'$.

Suppose $\ket{\psi}$ in $\mc{C}$ was the message that was transmitted
and the received message was $\ket{\phi}=U_\Vue V_{\sigma_\Sm\Vve}\ket{\psi}$ for
some unknown $\Vue$ and $\Vve$ in $\F[p]^n$.  It is sufficient to find
$\Vue$ and $\Vve$ to recover the actual message: given $\Vue$ and
$\Vve$ we apply the operator $V_{\sigma_\Sm \Vve}^\dagger
U_{\Vue}^\dagger$ to $\ket{\phi}$. The following proposition plays an
important role in finding $\Vue$ and $\Vve$ efficiently.

\begin{proposition}\label{isotropy_poly_computation_proposition}
For any $( \Va,\Vb ) \in S$, we can efficiently compute the polynomial $\Va(X)\Vve(X^{-\Sm})
  -\Vb(X) \Vue(X^{-\Sm}) \mod X^n-1$
\end{proposition}
\begin{proof}

  We need to compute the following polynomial.
  \begin{equation}\label{eqn:recover-poly}
    \Va(x)\Vve(X^{-\Sm}) - \Vb(X)\Vue(X^{-\Sm})= \sum_{i=0}^{n-1} \Symplectic{(\SH^{i}\Va,\SH^{i}\Vb)}{(\Vue,\Vve)}_{\sigma_m}X^{mi} \mod X^n-1
  \end{equation}
  We would recover one coefficient at a time. 
  Recall that the sent message $\ket{\psi}$ is
  stabilized by $W_{\Va,\Vb}'$ and it is easy to verify the
  commutation relation
  \[
  W_{\Va,\Vb}' U_{\Vue}V_{\sigma_m\Vve} =
  \omega^{\Symplectic{(\Va,\Vb)}{(\Vue,\Vve)}_{\sigma_m}}  U_{\Vue}V_{\sigma_m\Vve}
  W_{\Va,\Vb}'.
  \]
  Hence the received vector $\ket{\phi}=U_\Vue V_{\sigma_\Sm\Vve}
  \ket{\psi}$ is an eigen vector of $W_{\Va,\Vb}'$ with eigen value
  $\omega^{\Symplectic{(\Va,\Vb)}{(\Vue,\Vve)}_{\sigma_m}}$. Using the phase
  estimation algorithm~\cite[5.2]{nielsen2002quantum}, we extract the
  inner product $\Symplectic{(\Va,\Vb)}{(\Vue,\Vve)}_{\sigma_m}$ without
  modifying the received state $\ket{\phi}$. This recovers the
  constant term of the polynomial

  We repeat the above procedure with the Weyl Operator $W'_{\SH^i\Va,\SH^i\Vb}$
  to compute the coefficient $\Symplectic{(\SH^{i}\Va,\SH^{i}\Vb)}{(\Vue,\Vve)}_{\sigma_m} $.
  Each of these phase estimations
  gives us an additional coefficient of the polynomial. This allows
  us to recover the polynomial in Equation~\ref{eqn:recover-poly}
  after $n$ phase estimations.
 \end{proof}

  Notice that if the pair $(\Vue,\sigma_m\Vve)$ is of joint weight
  less than $\tau$, the polynomial $\Vue + \eta \Vve$ as a polynomial
  in $\PR[\eta]$ will have at most $2\tau$ non-zero coefficients. The
  main idea is to use 
  Proposition~\ref{isotropy_poly_computation_proposition} to recover
  the polynomial $\Vue + \eta \Vve$ modulo the generator polynomial
  $h(X,\eta)$. Then using the classical Berlekamp-Massey-Welch algorithm
  we recover $\Vue + \eta \Vve$. This is formalized in the following 
  proposition.

\begin{theorem}\label{error_correction_theorem}
  Let $S$ be $\sigma_\Sm$-isotropic ideal of $\PR(\eta)$. Let $g(X)$, $h(X,\eta)$
  be the polynomial satisfying the properties in Theorem \ref{general_construction_theorem}.
  Let $h(X,\eta)$ be of BCH distance $4\tau +1$.
  There exists an efficient quantum algorithm that corrects errors of joint weight $\tau$
\end{theorem}
\begin{proof}

From the proof of Theorem \ref{general_construction_theorem}
(Section \ref{proof_construction_theorem}), we know that
there exists a polynomial $\Va(X) \in \PR$ such that
$\Vg(X) + \eta \Va(X) \Vg(X)$ belongs to the ideal $S$. By abusing the notation,
 let $(\Vg,\Va\Vg) \in S$.
For $(\Vue,\Vve) \in \FFpn$,
 \begin{equation*}
    r'(X)=\Vg(X)\Vve(X^{-\Sm}) - \Va(X)\Vg(X)\Vue(X^{-\Sm})
\end{equation*}
By Proposition \ref{isotropy_poly_computation_proposition}, we can compute $r'(X)$ efficiently.

Divide both side by $g(X)$ and then take modulo $h(X,\eta)$. 
From the proof of Theorem \ref{general_construction_theorem}
(Section \ref{proof_construction_theorem}, Proposition \ref{h_factor_proposition})
we know that $\Va(X) = \eta^p \mod h(X,\eta)$. Thus,
\begin{equation*}
    r(X)=\Vve(X^{-\Sm}) - \eta^p \Vue(X^{-\Sm}) \mod h(X,\eta)
\end{equation*}
If joint weight  $wt\left(\Vue,\sigma_\Sm\Vve \right)$ is at most $\tau$
then the joint weight of $wt(\Vue,\Vve)$ is at most $2\tau$. 
Notice that the joint weight $wt\left( \sigma_\Sm\Vue,\sigma_\Sm\ \Vve \right)$ is
 equal to the joint weight $wt(\Vue,\Vve)$. Thus,  we could use 
 Berlekamp-Massey-Welch algorithm (Theorem \ref{berlekamp_theorem}) to
 compute $\Vue,\Vve$.

 \end{proof} 

\section{Explicit examples}\label{section_code_examples}

The Table \ref{constructed_codes_table} shows the codes constructed
over $\F[2]$ based on Theorem \ref{general_construction_theorem}. Fix a primitive
$n$-th root of unity $\beta$. The table uses the following notation:
$g_i$ (respectively $h_i$) is the irreducible factor of $X^n - 1$ over
the field $\F[2]$ (respectively $\F[2^2]$) with $\beta^i$ as one of
its roots.

\begin{table}[!ht]
\centering
\caption{Explicit examples of codes over $\F[2]$}
 \begin{tabular}{||c | c | c |c | c| c | c |c | c||}
 \hline
 \multirow{2}{*}{$n$} &\multirow{2}{*}{$k$} &
\multicolumn{2}{c|}{Factors}& Consecutive Root& \multicolumn{2}{c|}{Theorem \ref{thm:error_correction_sigma_stabilizer}} & \multicolumn{2}{c||}{Brute Force}\\ \cline{3-4}\cline{6-9}
                      &  & $g(X)$& $h(X,\eta)$ & of $h(X,\eta)$ & Detect & Correct & Detect & Correct \\ [0.5ex]

\hline\hline
 $5$ & $1$ & $g_0$ & $h_2$ & $\beta^2,\beta^3$  & 1 & 0 & 1 & 0 \\ \hline
 $9$ & $1$ & $g_0$ & $h_2 h_6$ & $\beta^5,\beta^6$  & 1 & 0 & 1 & 0 \\ \hline
 $11$ & $1$ & $g_0$ & $h_1$ & $\beta^3,\beta^4,\beta^5$  & 1 & 0 & 2 & 1 \\ \hline
 $13$ & $1$ & $g_0$ & $h_2$ & $\beta^5,\cdots,\beta^8$  & 2 & 1 & 4 & 2 \\ \hline
 $15$ & $1$ & $g_0$ & $h_1h_5h_6h_7$ & $\beta^4,\cdots,\beta^7$  & 2 & 1 & 3 & 1 \\ \hline
 $15$ & $5$ & $g_0g_7$ & $h_1h_5h_6$ & $\beta^4,\beta^5,\beta^6$  & 1 & 0 & 2 & 1 \\ \hline
 $15$ & $9$ & $g_0g_3g_7$ & $h_1h_5$ & $\beta^4,\beta^5$  & 1 & 0 & 1 & 0 \\ \hline
 $17$ & $1$ & $g_0$ & $h_2h_6$ & $\beta^6,\cdots,\beta^{11}$  & 3 & 1 & 5 & 2 \\ \hline
 $19$ & $1$ & $g_0$ & $h_1$ & $\beta^4,\cdots,\beta^{7}$  & 2 & 1 & 3 & 1 \\ \hline
 $21$ & $13$ & $g_0g_3g_5g_9$ & $h_1$ & $\beta^7,\beta^{8}$  & 1 & 0 & 1 & 0 \\ \hline
 $25$ & $1$ & $g_0$ & $h_1h_5$ & $\beta^4,\beta^5,\beta^{6}$  & 1 & 0 & 2 & 1 \\ \hline
 $27$ & $7$ & $g_0g_3$ & $h_1h_9$ & $\beta^9,\beta^{10}$  & 1 & 0 & 1 & 0 \\ \hline
 $29$ & $1$ & $g_0$ & $h_1$ & $\beta^4,\cdots,\beta^7$  & 2 & 1 & 3 & 1 \\ \hline

% \hline\hline
% 5 & 1 & 1& 3 & 2 \\
% \hline
% 9 & 1 & 1& 3 & 2 \\
% \hline
% 11 & 1 &1 & 4 & 3 \\
% \hline
% 13 & 1 & 2& 5 & 5 \\
% \hline
% 15 & 1 & 2&  5 & 4 \\
% \hline
% 15 & 5 & 1& 4 & 3 \\
% \hline
% 15 & 9 & 1& 3 & 2 \\
% \hline
% 17 & 1 & 3& 7 & 6 \\
% \hline
% 19 & 1 & 2& 5 & 4 \\
% \hline
% 21 & 13 &1& 3 & 2 \\
% \hline
% 25 & 1 & 1& 3 & 3 \\
% \hline
% 27 & 7 &1 & 3 & 2\\ [1ex]
\end{tabular}
\label{constructed_codes_table}
\end{table}

\ppk{If the brute force distance is 4 then the code can only correct 1 error not 2 as given in the table. Please clarify this point.}

\tgandhi{ The value 4 is about error detection and not the distance. In previous version we had written distance values under the heading of detect which was incorrect. }

The  Table \ref{constructed_codes_table}  shows distance
based on Theorem \ref{thm:error_correction_sigma_stabilizer} as well as
by brute force computation. In our experiments we found the minimum 
distance of $\C{S}^{\sigma}$ to be almost equal to the minimum distance of 
$\C{S}$ ( which could be seen in the table). This corroborates our claim
 that the distance bound of Therorem \ref{thm:error_correction_sigma_stabilizer}
is a bit conservative.

In order to better analyze the performance of our codes we have 
performed simulation of the effect of depolarizing channel on some of the codes presented in Table
 \ref{constructed_codes_table}. These simulations show that there exists a threshold probability below which 
increasing the dimension improves the performance of our codes. 
More details are given in the appendix~\ref{sec_simulations}.

\section{Conclusion}

The main theme of this article is to construct stabilizer codes based
on alternate symplectic forms. Any two (full rank) symplectic forms
are equivalent in the sense that the associated Weyl operators
form a basis set for the error space and hence can mathematically
model all quantum operations on the relevant Hilbert space. Modifying
the symplectic form therefore is clearly \emph{not just restricted} to
cyclic codes. However, if we need to get meaningful bounds on the
distance, these changes needs to be balanced carefully. In the context
of cyclic linear stabilizer code, symplectic forms of the kind
$\Symplectic{\cdot}{\cdot}_{\sigma_m}$ were the only ones that gave us
enough control to carry out our constructions and get nontrivial
distance bounds at the same time. A future line of work would be to
extend some of the ideas here to general stabilizer codes. We believe
would lead to some interesting examples of quantum codes.

The equivalence of symplectic forms means that our construction could
as well be carried out by considering the set $S^{\sigma_m}$ under the
standard isotropy condition.  However, notice that the set
$S^{\sigma_m}$ as opposed to $S$ is \emph{not cyclic} and hence the
efficient decoding algorithms that we have will not be apparent in the
setting of the standard symplectic forms. The reason for this anomaly
is that properties like distance and cyclicity are \emph{not
  preserved} under a basis change. Therefore, visualizing this code as
the subspace $S$ as opposed to $S^{\sigma_m}$ is crucial. This is what
sets the codes apart from other constructions of codes for similar
lengths. In general, decoding is an intractable problem even for
classical codes.

\bibliographystyle{plainurl}
\bibliography{../QuantumAuxFiles/qecc}

\appendix
\section{Appendix}
%%\subsection{Modified symplectic form}

We will setup some preliminaries before proceeding with the proofs.

For a given polynomial $\Vg(X)$ in $\PR$, the polynomial $\Vg(X^{-\Sm})$ plays an important role in the $\sigma_\Sm$-isotropy condition. 
%In certain cases, the polynomial $\Vg(X^{-\Sm})$ can be expressed in terms of the polynomial $\Vg(X)$.
If there exists $t$ such that $n \divides  p^t + \Sm$, then $X^{-\Sm }=X^{p^t}$ (If $m=-1$, then we will choose $t=0$). Therefore, for any polynomial $\Va \in \PR$, 

\begin{equation}\label{frobenius_equation}
 \Va(X^{-\Sm}) = \Va(X)^{p^t} \mod X^n-1 .
\end{equation}

%\begin{proof}
%This proposition is quite straightforward to show. When $\Sm=-1$ let $t=0$. In the case  $n \divides p^t+\Sm$ let $\alpha=t$, we have $-\Sm=p^t \mod n$. Hence, $\Va(X^{-\Sm})=\Va(X^{p^t})=\Va(X)^{p^t}$.
% \end{proof} 

%By Equation \ref{frobenius_equation}, 

The $\sigma_\Sm$-isotropy (Equation~\ref{polynomial_isotropy_condition}) becomes
\[ \Va(X)\Vd(X)^{p^t} - \Vb(X)\Vc(X)^{p^t} \]

%Based on the above $\sigma_\Sm$-isotropy condition, the proof of most of the theorems are similar to the proof of theorems in the work of Dutta et. al.  \cite{DuttaK2011cyclic,dutta2012exploiting}.

%\subsection{Simultaneously cyclic}

For a given simultaneously cyclic subspace $S$ of $\mathbb{F}_p^n \times \mathbb{F}_p^n$, let's define the following two set.
\begin{equation}
G =  \set*{\Va \given \text{exists $\Vb$ such that } (\Va,\Vb) \in S}
\end{equation}
\begin{equation}
H =  \set*{\Vb \given (\mathbf{0},\Vb) \in S}
\end{equation}

As $S$ is simultaneously cyclic, both $G$, $H$ are cyclic code in $\F[p]^n$. As seen earlier,
they could be thought of as an ideal in $\PR$. Thus they are generated by a factor of $X^n-1$.
Let $\Vg(X)$ and $\Vh(X)$ be the factors of $X^n-1$ that generates $G$ and $H$ as the ideals of $\PR$ respectively.

It is easy to see that any simultaneously cyclic subspace of $\FFpn$ could be be described by three polynomials over $\PR$.

\begin{lemma}\label{simultaneously_generator_lemma}
Every simultaneously cyclic subspace $S$ of $\FFpn$ can be described by three polynomials 
$\Vg(X)$, $\Vh(X)$ which are factor of $X^n-1$ over $\F[p]$ and $\Vf(X) \in \PR$ such that $(\Vg,\Vf) \in S$ and $(\mathbf{0},\Vh) \in S$.
\end{lemma}

In Lemma \ref{simultaneously_generator_lemma} if $\Vh(X)=0 \mod X^n -1$ then such a subspace is known as \emph{uniquely cyclic}.

\begin{lemma}\label{uniquely_cyclic_lemma}
Let $\Sm=-1$ or $n$ divides $p^t+\Sm$. Let $S$ be a simultaneously cyclic, $\sigma_\Sm $-isotropic and $\mathbb{F}_{p^2}$ linear subspace then $S$ is uniquely cyclic.
\end{lemma}
\begin{proof}
The proof of this lemma is similar to the proof of Lemma $IV.2$ in \cite{DuttaK2011cyclic}. For completeness we state it here.

The Lemma \ref{simultaneously_generator_lemma} states that $S$ could be expressed by three polynomials $\Vg(X)$, $\Vf(X)$, and $h\left( X \right)$.
 The polynomials $\Vg(X)$, $\Vh(X)$ are factors of $X^n-1$ and $\Vf(X)$ is a polynomial in $\PR$ such that the elements $(\Vg,\Vf)$, $(0,\Vh)$ are in $S$.

To show that $S$ is an uniquely cyclic subspace we need to show that $\Vh(X)$ is a multiple of $X^n -1$.
 We would show this in two steps. 
First, we would show $\Vg(X) \divides \Vh(X)$. This proof is similar as in Lemma $IV.2$\cite{DuttaK2011cyclic}(\cite[5.11]{dutta2012exploiting}).
Now we would show $\frac{X^n-1}{\Vg(X)} \divides \Vh(X)$.

 The elements $(\Vg,\Vh)$, $(0,\Vh)$ belongs to $S$ and it is $\sigma_\Sm $-isotropic. Thus,  from the $\sigma_\Sm$-isotropy
condition (\ref{polynomial_isotropy_condition}) between $(\Vg,\Vf)$ and $(0,\Vh)$,  we have the following.
\begin{align*} 
&&  \Vg(X^{-\Sm })\Vh(X)&=0 \mod X^n-1&&  \\
&&  \Vg(X)^{p^t}\Vh(X)&=0 \mod X^n-1 && \textrm{(From Equation \ref{frobenius_equation})}\\
&& \Vh(X)&=  0 \mod \frac{X^n-1}{\Vg(X)}&&
\end{align*}
Thus $\frac{X^n-1}{\Vg(X)} \divides \Vh(X)$.

 \end{proof} 

\subsection{Proof of Theorem \ref{general_construction_theorem} } \label{proof_construction_theorem}

%Let $c(Y)=c_0+c_1Y-Y^2$ be an irreducible polynomial of degree 2 over $\mathbb{F}_p\left[ Y \right]$. Let $\mathbb{F}_p(\eta)=\mathbb{F}_p[Y]/c(Y)$ be the quadratic extension over $\mathbb{F}_p$ where $\eta$ and $\eta^p$ are the roots of polynomial $c(Y)$ in it.
We would characterize the $\sigma_\Sm$-isotropic ideal $S$ of $\PR(\eta)$.
From Lemma \ref{simultaneously_generator_lemma}  and Lemma \ref{uniquely_cyclic_lemma}, 
we know that such an ideal $S$ could be expressed by two polynomials $\Vg(X)$ and $\Vf(X)$. 
The polynomials $\Vg(X)$ is a factor of $X^n-1$ and $\Vf(X)$ is a polynomial in $\PR $ such that $(\Vg,\Vf) \in S$.

Any element of $S$ can be expressed as $\Vb(X)\Vg(X) + \eta \Vb(X)\Vf(X)$, for some $\Vb(X)$ in $\PR$.
 Since, we have $\Vg(X)+\eta \Vf(X)$ in $S$, by $\mathbb{F}_{p}\left( \eta \right)$ linearity,
 we have  $\eta(\Vg(X)+\eta \Vf(X))$ in $S$. There exists a polynomial $\Va(X) \in \PR$ such that we get the following.
\begin{align}
&&\eta\left( \Vg(X) + \eta \Vf\left( X \right) \right)=& \Va(X)\left( \Vg(X) + \eta \Vf(X) \right) &\mod X^n-1 \label{gen_equation}
\end{align}

Compare the coefficient of $\eta$ in Equation \ref{gen_equation}, we have the following.

\begin{align}
&&\Vf(X) =&c_0^{-1}\Va(X)\Vg(X) &\mod X^n -1 \label{f_equation}\\
&&\mu(\Va(X)) =& 0 &\mod \frac{X^n-1}{\Vg(X)} \label{c_a_equation}
\end{align}

Recall that $\eta$ is a root of an irreducible quadratic polynomial $\mu(Y)=Y^2-c_1 Y -c_0$.
Let $r(X)$ be any irreducible factor $r(X)$ of $\frac{X^n-1}{\Vg(X)}$. 
Then Equation \ref{c_a_equation} implies that the $\Va(X) \mod r(X)$ is a root of $\mu(X)$. 
An immediate consequence of this is the following proposition.
\begin{proposition}\label{even_degree_proposition}
$\Vg(X)$ contains all the odd degree factors.
\end{proposition}
\begin{proof}
%    The proof of this proposition could be derived from Theorem 5.15 \cite[Chapter 5]{dutta2012exploiting}.

%From \ref{c_a_equation}, we know $\Va(X) \mod r(X)$ is a root of $\mu(Y)$ over it. 
The field extension $\F[p][X]/r(X)$ contains the root of the polynomial $\mu(X)$. 
Thus it contains $\F[p^2]$ as a subfield. This implies $r(X)$ has to be of even degree.
\end{proof}

By Proposition \ref{even_degree_proposition}, we know that $r(X)$ is of even degree.
Thus each $r(X)$ factorizes as $r'(X,\eta)r'(X,\eta)^p$ over $\F[p^2]$.
\begin{proposition}\label{h_factor_proposition}
$r'(X,\eta) \divides h(X,\eta)$ if and only if $r'(X,\eta)^p \ndivides h(X,\eta)$
\end{proposition}
\begin{proof}[Sketch]
It could be shown that $\Vh(X,\eta)=gcd\left( \frac{x^n-1}{\Vg(X)},1+\eta c_0^{-1} \Va\left( X \right) \right)$.
By Equation \ref{c_a_equation} we know $\Va(X) \mod r'(X,\eta)$ is either $\eta$ or $\eta^p$.
We could also show that  $a(X)=\eta^p \mod r'(X,\eta)$
if and only if $ a(X)=\eta \mod r'(X,\eta)^p $.
%Thus $\Va(X) = \eta^p $ modulo either $r'(X,\eta)$ or $r'(X,\eta)^p$.
Thus $r'(X,\eta) \divides h(X,\eta)$ if and only if $\Va(X,\eta) = \eta^p \mod r'(X,\eta)$.
%Thus only one of the factors $r'(X,\eta)$ and $r'(X,\eta)^p$ divides $h(X,\eta)$.
The details of the proof could be derived from Theorem 5.15 \cite[Chapter 5]{dutta2012exploiting}.
 \end{proof} 

When $m=-1$, $t$ is the order of $p$ in $\ZZ_n$. However, when 
the order of $p$ in $\ZZ_n$ is odd there are no even degree factors of
$X^n-1$ over $\F[p]$. From Proposition \ref{even_degree_proposition} we have 
$g(X)$ is a multiple of $X^n-1$. Thus for ideal to be 
non-trivial $t$ has to be even. 
This completes the proof for the case when $m=-1$. However, for $m \neq -1$ a little more is required.

As the $S$ is $\sigma_\Sm $-isotropic subspace, the $\sigma_\Sm $-isotropy condition of Equation \ref{polynomial_isotropy_condition} for $(\Vg,\Vf)$ with itself yields the following.

\begin{align*}
%\Vg(X)\Vf(\sigma(X^{-1}))&=&&\Vf(X)\Vg(\sigma(X^{-1}) ) && \mod X^n-1 \nonumber\\
&&\Vg(X)\Vf(X^{-\Sm })=&\Vf(X)\Vg(X^{-\Sm }) &\mod X^n-1 \\
&&\Vg(X)\Va(X^{-\Sm })\Vg(X^{-\Sm })=&\Va(X)\Vg(X)\Vg(X^{-\Sm }) &\mod X^n-1 \\
\end{align*}
From Equation \ref{frobenius_equation}, we have,
\begin{align}
&&\Va(X)^{p^t}\Vg(X)^{p^t + 1}=&\Va(X)\Vg(X)^{p^t+1} & \mod X^n-1 \nonumber \\
&&\Va(X)^{p^t}=&\Va(X) &\mod \frac{X^n-1}{\Vg(X)} \label{a_equation} 
\end{align}

Its easy to see when $m=-1$ Equation \ref{a_equation} 
is trivially satisfied. Hence does not impose any further constraints in such a case.
However, for $m\neq -1$ assuming $t$ to be odd, the Equation \ref{c_a_equation} and 
Equation \ref{a_equation} leads to contradiction. Details could be 
derived from Theorem 5.15 \cite[Chapter 5]{dutta2012exploiting}.
Thus $t$ has to be even.

This completes the proof of Theorem \ref{general_construction_theorem}.

\subsection{Proof of Proposition \ref{prop:centralizer-generator}} \label{proof_centralizer-generator}

 Let $S$ be a $\sigma_\Sm$-isotropic ideal of $\PR(\eta)$ as in  Theorem \ref{general_construction_theorem}.
We need to show that $h(X,\eta)$ maps to the $\sigma_\Sm$-centralizer of $S$.
At first, we compute the size of $S$. This would determine the size
 of $\sigma_\Sm$-centralizer $\C{S}$.

By Lemma \ref{uniquely_cyclic_lemma}  we know that $S$ is expressed 
by two polynomials $g\left( X \right)$ and $h\left( X,\eta \right)$.
From Section \ref{proof_construction_theorem} we know that there 
exists a polynomial $\Va(X) \in \PR$ and  $\Vg(X)$ is a factor of $X^n-1$ over $\F[p]$
such that 
$\Vg(X)+\eta\Va(X)\Vg(X)$ generates $S$.
The following proposition formalizes the dimension of the $S$ in terms of the polynomial $g(X)$.
\begin{proposition}
\cite[Theorem 5.24]{dutta2012exploiting}\cite[Theoverm V.7]{DuttaK2011cyclic} The dimension of the $S$ as a subspace is  $n-\deg \Vg(X)$.
\end{proposition}
%\begin{proof}
%From \ref{general_construction_theorem}, we know there exists a polynomial $\Va(X) \in \PR$  such that $S$ is a $\sigma_\Sm$-isotropic ideal in $\PR(\eta)$. The elements of $S$ are of the form $(\Vu,\Vv)$, where $\Vu$ corresponds to a polynomial $\Vu(X)=w(X)\Vg(X), w(X) \in \mathbb{F}_{p}\left[ X \right]/\left< \frac{X^n-1}{\Vg(X)}\right>$ and similarly $\Vv$ corresponds to $v(X)=\Va(X)w(X)\Vg(X)$. Since $ \mathbb{F}_{p}\left[ X \right]/\left< \frac{X^n-1}{\Vg(X)}\right>$ is of size $p^{n-\deg \Vg(X)}$, the dimension of $S$ as subspace is $n- \deg \Vg(X)$
% \end{proof} 

%As $\C{\mc{S}}$, a $\sigma_\Sm$-centralizer, is an ideal in $\PR(\eta)$. We identify the generating polynomial of $\C{\mc{S}}$.

Let $\Va(X) \in \PR$ be the polynomial as in the proof of Theorem 
\ref{general_construction_theorem} (Section \ref{proof_construction_theorem}, Equation \ref{gen_equation}). 
For such a fixed polynomial $\Va(X)$, let $Z$ be the set defined as follows.

\begin{equation}
Z=\set*{
 \left( \Vu,\Vu\Va + \Vv' \right) \given
 \begin{aligned}
 & \forall \Vu(X) \in \PR, \\
 & \forall \Vv(X) \in \F[p]\left[ X \right]/\left( \frac{X^n-1}{\Vg(X)} \right) ,\Vv'(X)=\Vv(X)\left(\frac{ X^n-1}{\Vg(X)}\right)
  \end{aligned}
}
\end{equation}

\begin{proposition}
The set $Z$ is the $\sigma_{\Sm}$-centralizer of $S$.
\end{proposition}
\begin{proof} 

First, we need to show that any element $\left( \Vu,\Vv \right) \in Z$ satisfies the 
$\sigma_\Sm$-isotropy condition with every element of $S$. 
It is enough to show that the element $(\Vu,\Vv)$ is $\sigma_\Sm$-isotropic with $\left( \Vg,\Va\Vg \right)$.
From the polynomial form of the 
$\sigma_\Sm$-isotropy condition (\ref{polynomial_isotropy_condition}), we have
\begin{align*}
%    \Vu(X)a(\sigma_\Sm(X^{-1}))\Vg(\sigma_\Sm(X^{-1})) &=& v(X) \Vg(\sigma_\Sm(X^{-1})) \mod X^n-1\\
&&   \Vu(X)\Va(X^{-\Sm})\Vg(X^{-\Sm})=& \left( \Vu(X)\Va(X) + t(X) \frac{X^n-1}{\Vg(X)}\right) \Vg(X^{-\Sm}) & \mod X^n-1\\
\end{align*}
 From Equation \ref{frobenius_equation}
\begin{align*}
&&    \Vu(X)\Va(X)^{p^t}\Vg(X)^{p^t} =& \left( \Vu(X)\Va(X) + t(X) \frac{X^n-1}{\Vg(X)}\right) \Vg(X)^{p^t} &\mod X^n-1\\
&&    \Vu(X)\Va(X)^{p^t}\Vg(X)^{p^t} =&  \Vu(X)\Va(X)\Vg(X)^{p^t} &\mod X^n-1  \\
&&    \Vu(X)\Va(X)^{p^t}           =&  \Vu(X)\Va(X) &\mod \frac{X^n-1}{\Vg(X)}  \\
\end{align*}
Since $\Va(X)^{p^t} = \Va(X) \mod (X^n-1)/\Vg(X)$ (from Section \ref{proof_construction_theorem}), the above equation is satisfied. The cardinality of the set $Z$ is $p^{n+\deg g}$ which is same as the cardinality of $\sigma_\Sm$-centralizer $\C{S}$. Hence, we have shown that  $Z$ is the centralizer of $S$.

Now we need to show $\sigma_\Sm$-centralizer $\C{S}$ maps to $h(X,\eta)$. 
Let $(\Vu,\Vv) \in \FFpn$  maps to an element $\Vu(X) + c_0^{-1}\eta \Vv(X) \in \PR(\eta)$.
It is easy to see that the joint weight of an element remains unchanged under this mapping.
Any element of $\left( \Vu,\Vv \right) \in Z$ maps to $\Vu(X)\left(1+c_0^{-1}\eta \Va(X)\right) + c_0^{-1}\eta \left(t(X)\left( \frac{X^n-1}{\Vg(X)} \right)\right)$.
As $\Vh(X,\eta)=gcd\left( \frac{x^n-1}{\Vg(X)},1+\eta c_0^{-1} \Va\left( X \right) \right)$ (from Section \ref{proof_construction_theorem}), $\Vh(X,\eta)$ divides it.
This completes the proof.
 \end{proof} 

%\subsection{Proof of Theorem \ref{error_correction_theorem}} \label{proof_error_correction_theorem}

%For the proof of the following theorem, we require the following classical decoding algorithm.
%\begin{theorem}[Berlekamp-Massey-Welch]
%Let $h(X,\eta)$ be a factor of $X^n-1$ over $\F[p](\eta)$ with $d-1$ consecutive roots. Thus, it's BCH distance is $d=2\tau+1$.  Let $r(X)=e(X) \mod g(X)$. If the Hamming weight of $e(X)$ is less than $\tau$ then there exists an efficient algorithm to find $e(X)$ given only $r(X)$.
%\end{theorem}

\section{Simulations} \label{sec_simulations}

Consider a depolarizing channel in which errors $X$, $Y$, $Z$ occur indpendently with 
probability $p$ and $I$ with probability $1-p$, for various values of $p$.
We simulated the performance of the maximum likelihood decoding algorithm for the codes $[[11,1,3]]$, 
$[[13,1,5]]$, and $[[17,1,6]]$ over this channel.

For the simulation, we send an all zero codeword through the depolarizing channel and compute the syndrome for the received word.
We employ the hard decision maximum likelihood decoder based on a look up table. The table stores the minimum symplectic weight vector for all syndromes. 
Based on this vector, a recovery operator is applied to complete the decoding. 

After decoding, if the resultant is a non zero codeword then we consider it as an error.
We repeat this process multiple times to compute the Quantum Block Error Rate(QBER). The QBER is defined as,
\[ QBER := \frac{\text{Total number of bits with error after decoding}}{\text{Total number of bits received}} .\]
In this simulation we do not take degeneracy into account.

For a general code, there is a trade-off between minimum distance and coding rate.
We consider similar trade-off between QBER and the coding rate. Our simulations show that below a certain probability (threshold probability), decreasing the 
coding rate results in a better QBER.

%It is a good measure of the performance of a code.
%It is know that for codes with linear distance the threshold probability exists.
%However, it also exists for codes with sub linear minimum distance. We demonstrate
%that such threshold probability exits for the codes constructed here.

We plot the QBER versus the depolarizing probability of $[[11,1,3]]$, 
$[[13,1,5]]$, and $[[17,1,6]]$ in Figure \ref{fig_11_qubit}, 
Figure \ref{fig_13_qubit}, and Figure \ref{fig_17_qubit} respectively. 
Figure \ref{fig_threshold_probability} shows the threshold probability which 
is indicated by the dashed line. It is the point where the QBER curves of $[[11,1,3]]$, 
$[[13,1,5]]$, and $[[17,1,6]]$  crosses over each other.

\begin{figure}
%\centering
\subcaptionbox{$[[ 11,1,3 ]]$ 
    \label{fig_11_qubit}%
}[0.5\textwidth]%
{ 
   \includegraphics[width=0.5\textwidth]{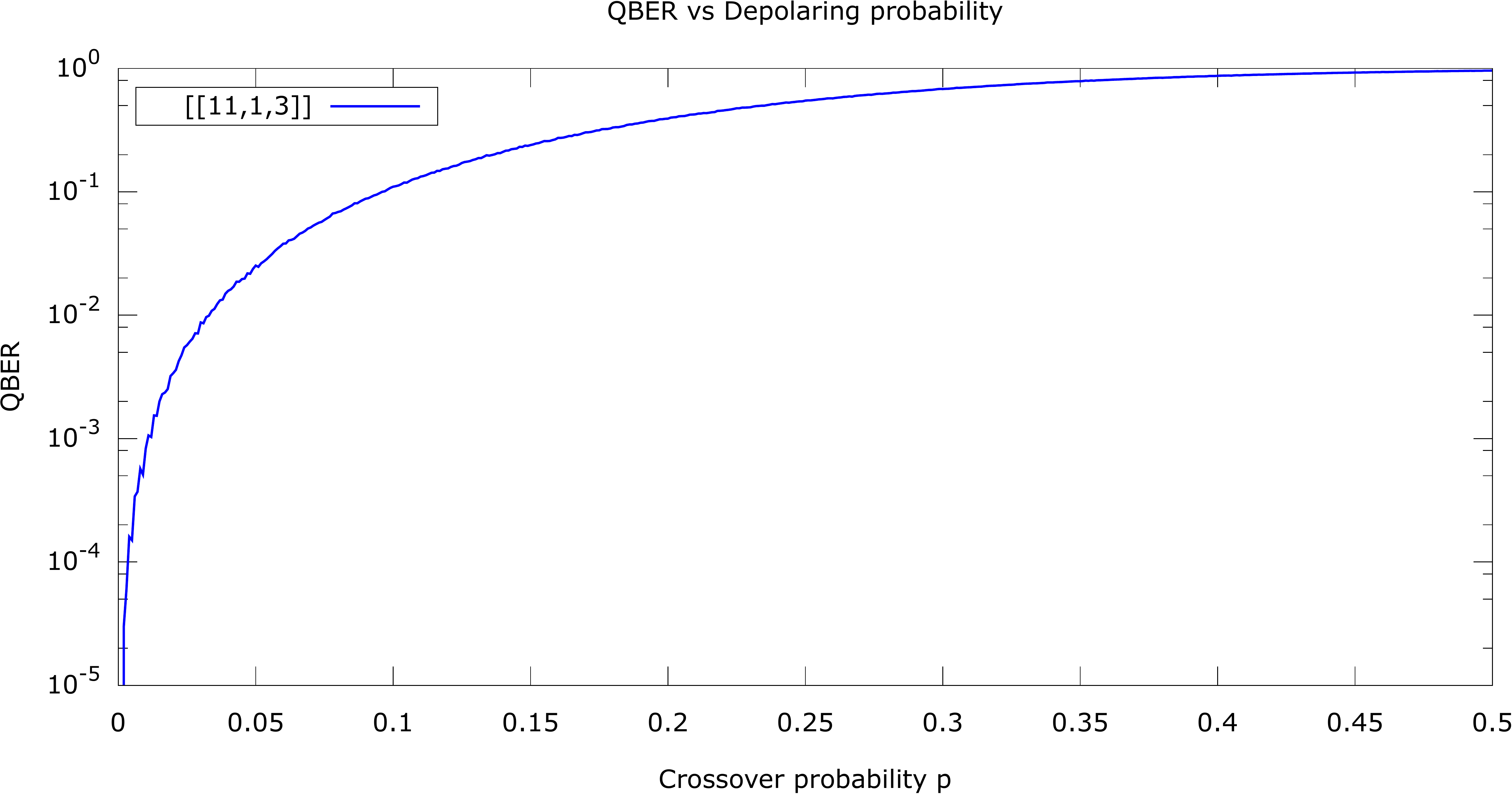}
}
%\hspace{0.1\textwidth} % seperation
\subcaptionbox{$[[ 13,1,5 ]]$ \label{fig_13_qubit}}
[0.5\textwidth]
{
  \includegraphics[width=0.5\textwidth]{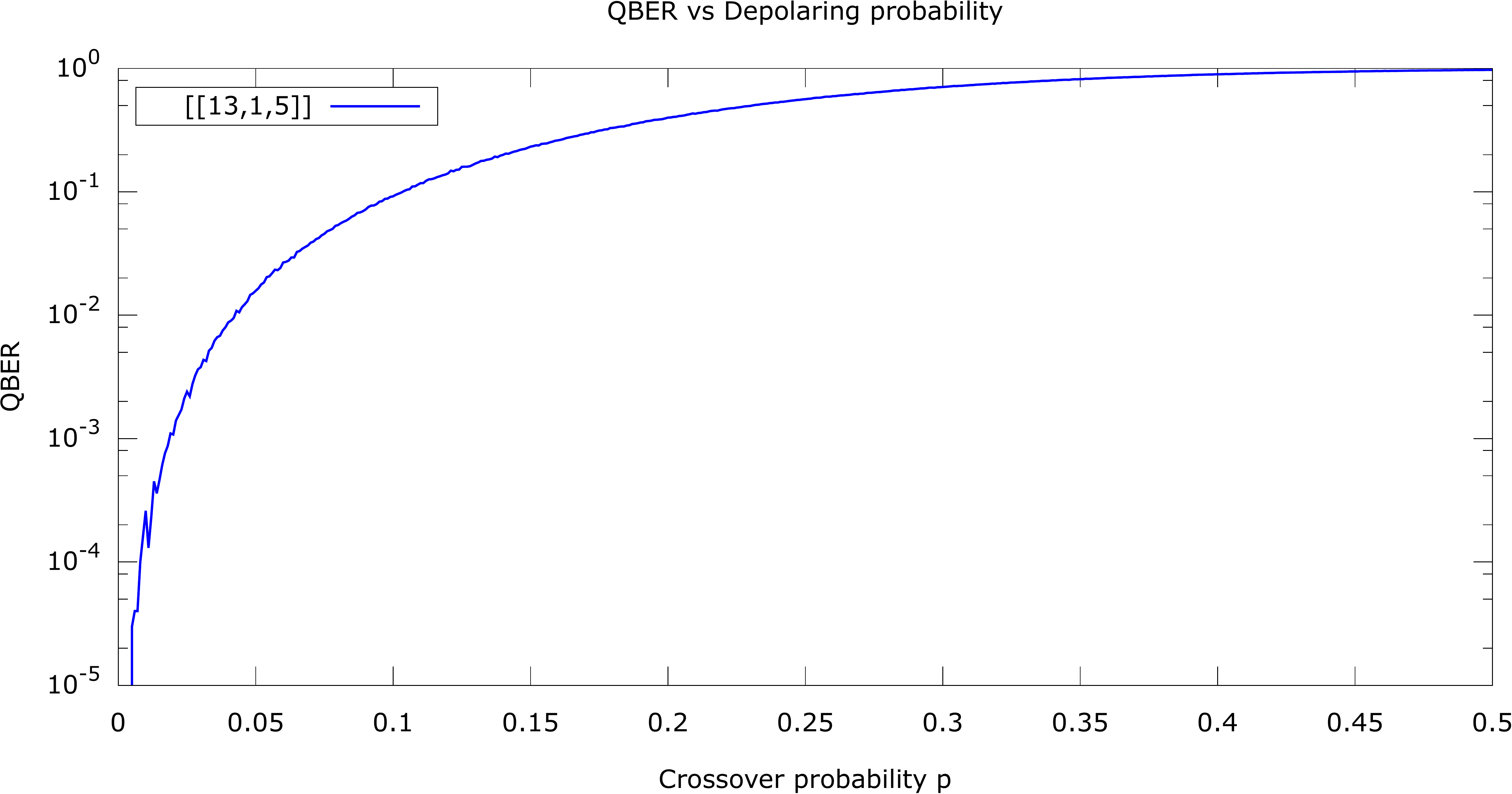}%
}
%\hspace{0.1\textwidth} % seperation
\subcaptionbox{$[[ 17,1,6 ]]$\label{fig_17_qubit}}
[0.5\textwidth]
{
    \includegraphics[width=0.5\textwidth]{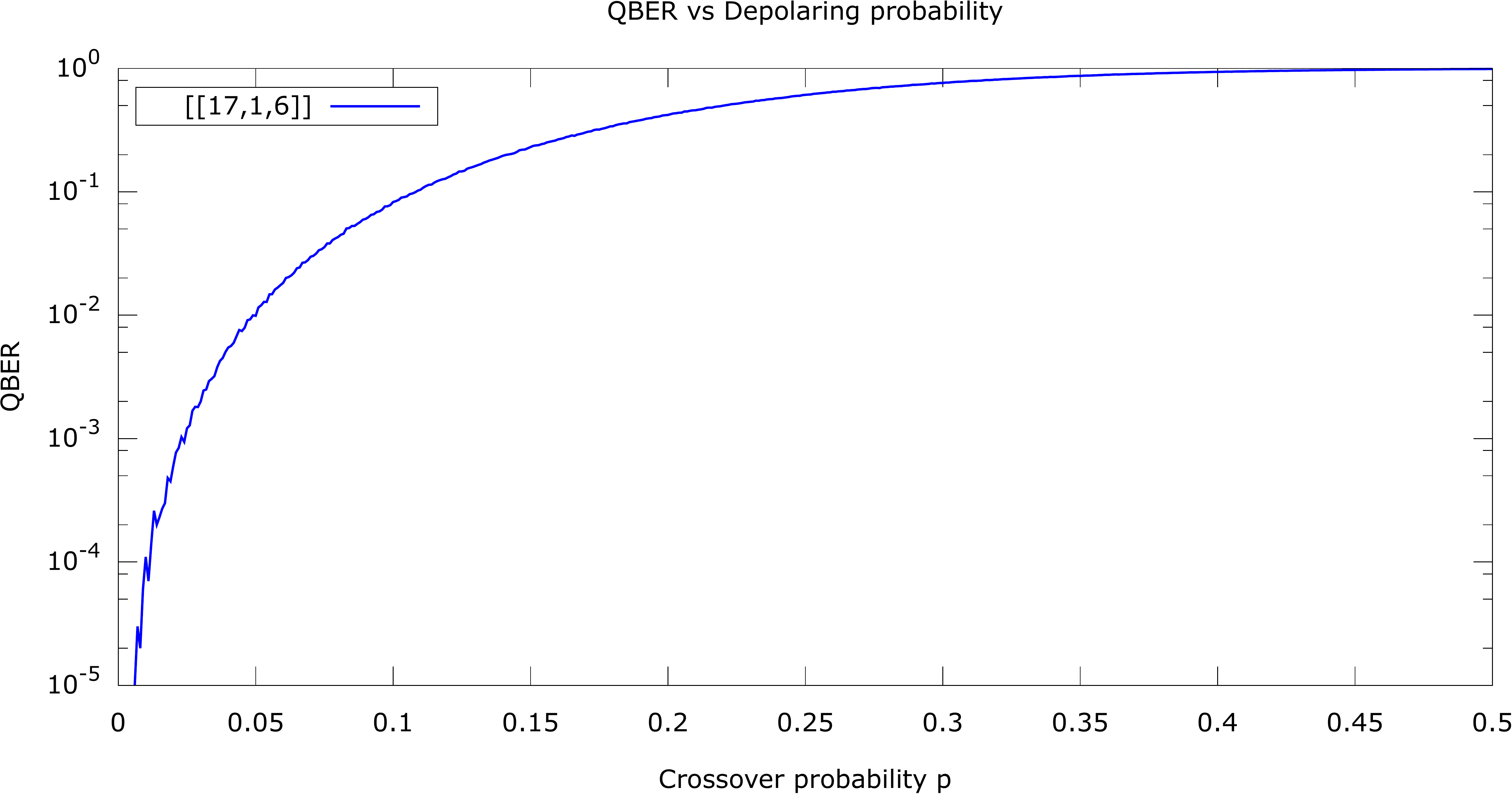}%
}
\caption{Crossover probability vs QBER for different codes}
\label{fig:label}
\end{figure}

\begin{figure}[ht!]
    \includegraphics[width=400pt]{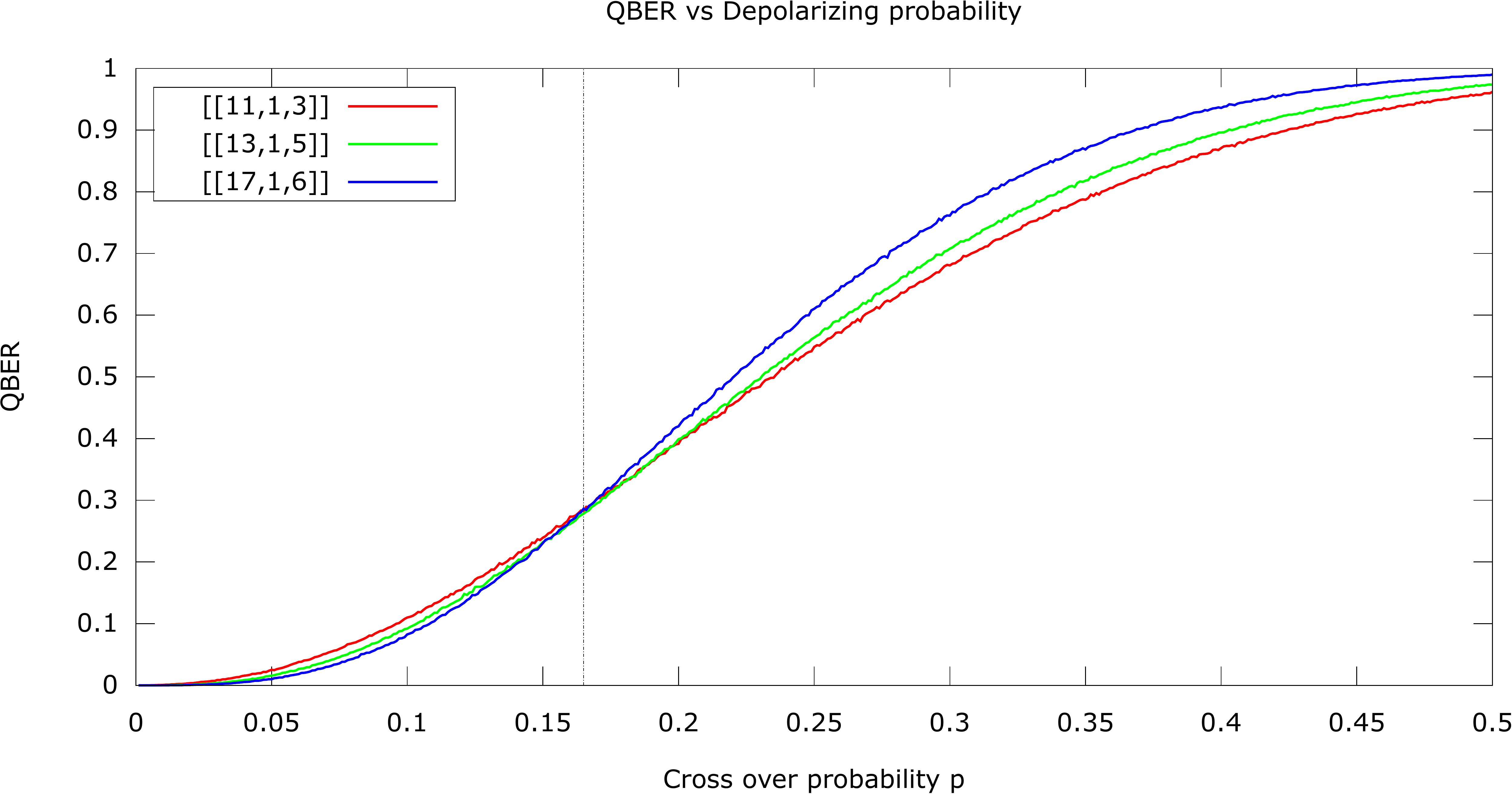}
    \caption{QBER performance as the coding rate decreases}
    \label{fig_threshold_probability}
\end{figure}

\end{document}